\documentclass[runningheads]{llncs}
\usepackage[T1]{fontenc}
\usepackage{graphicx}
\usepackage{amsmath}
\usepackage{amssymb}
\newtheorem{rem}{Remark}
\usepackage{booktabs}

\newcommand{\Real}{\mathbb{R}}
\newcommand{\eps}{\varepsilon}
\newcommand{\abs}[1]{\left\vert#1\right\vert}

\newcommand{\ra}{\rightarrow}
\newcommand{\set}[1]{\left\{#1\right\}}

\usepackage{makecell}
\usepackage{enumitem}
\usepackage{listings}

\usepackage{hyperref}
\usepackage{cite}

\usepackage[ruled,vlined,linesnumbered]{algorithm2e}
\usepackage{algorithmic}

\begin{document}
\title{Symbolic Reduction for Formal Synthesis of Global Lyapunov Functions}
\titlerunning{Symbolic Reduction for Formal Synthesis of Global Lyapunov Functions}
%
\author{
Jun Liu\inst{1}\orcidID{0000-0001-8762-2299} \and 
Maxwell Fitzsimmons\inst{2}\orcidID{0000-0002-3764-4542} 
}
\authorrunning{Liu and Fitzsimmons}
%
\institute{
University of Waterloo, Waterloo, Ontario N2L 3G1, Canada \\
\email{j.liu@uwaterloo.ca} \and
University of Waterloo, Waterloo, Ontario N2L 3G1, Canada \\
\email{mfitzsimmons@uwaterloo.ca}
}
\maketitle              
\begin{abstract}
We investigate the formal synthesis of global polynomial Lyapunov functions for polynomial vector fields. We establish that a sign-definite polynomial must satisfy specific algebraic constraints, which we leverage to develop a set of straightforward symbolic reduction rules. These rules can be recursively applied to symbolically simplify the Lyapunov candidate, enabling more efficient and robust discovery of Lyapunov functions via optimization or satisfiability modulo theories (SMT) solving. In many cases, without such simplification, finding a valid Lyapunov function is often infeasible. When strict Lyapunov functions are unavailable, we design synthesis procedures for finding weak Lyapunov functions to verify global asymptotic stability using LaSalle's invariance principle. Finally, we encode instability conditions for Lyapunov functions and develop SMT procedures to disprove global asymptotic stability. Through a series of examples, we demonstrate that the proposed symbolic reduction, LaSalle-type conditions, and instability tests allow us to efficiently solve many cases that would otherwise be challenging.

\end{abstract}

\section{Introduction}

Stability analysis is a fundamental problem in the formal modelling and verification of dynamical systems. Lyapunov methods \cite{lyapunov1992general,lasalle1961stability,zubov1961methods} remain among the most effective tools for establishing stability, particularly in determining whether a system will converge to a desired equilibrium point from a region beyond a neighborhood of the equilibrium. By constructing a Lyapunov function\textemdash a scalar function that decreases along system trajectories\textemdash one can certify stability without explicitly solving the system equations. Furthermore, any level set of a Lyapunov function provides a certified region of attraction. If a Lyapunov function is valid across the entire state space and its level sets are unbounded in all directions (a property known as radial unboundedness), then we can establish global attraction. Despite their effectiveness, finding suitable Lyapunov functions, especially global ones, is often challenging. Consequently, the construction of Lyapunov functions has received considerable attention \cite{giesl2007construction,giesl2015review}.

Sum-of-squares (SOS) programming has been the state-of-the-art approach for computing Lyapunov functions for polynomial systems, both globally \cite{papachristodoulou2002construction} and locally \cite{topcu2008local}. It allows for verifying the negative or positive (semi)-definiteness of polynomials by expressing them as a sum of squares of other polynomials, which can be efficiently managed using semidefinite programming (SDP). SOSTOOLS \cite{papachristodoulou2013sostools} is a widely used toolbox for SOS optimization, including Lyapunov function synthesis. While SDP provides a candidate solution in floating-point form, obtaining an exact symbolic representation requires postprocessing, typically by either filtering out smaller terms or retaining them, often in an ad hoc manner.

On the other hand, satisfiability modulo theories (SMT) solvers \cite{barrett2018satisfiability,de2008z3} offer an alternative approach with rigorous handling of exact arithmetic. SMT solvers enable direct verification and synthesis of Lyapunov functions by encoding stability conditions as logical formulas over real numbers. This approach allows for precise verification without numerical errors, making it particularly suitable for systems where symbolic precision is critical and SMT solving is feasible. SMT-based methods have recently gained traction as a promising tool for the formal synthesis of Lyapunov functions and controllers \cite{chang2019neural,ahmed2020automated,ravanbakhsh2015counter,ravanbakhsh2019learning,zhou2022neural,edwards2024fossil,liu2024tool,abate2017automated,liu2025physics}, especially in cases where SDP-based methods face limitations due to floating-point precision or the non-convex nature of the problems.

Linear programming (LP) is another framework that has been effectively used for the synthesis of Lyapunov functions. In \cite{sankaranarayanan2013lyapunov}, linear constraints are derived using interval evaluation of the polynomial form and "Handelman representations" for positive polynomials over polyhedral sets. Linear constraints can also be generated from simulations for a given template Lyapunov function, with counterexamples provided by an SMT solver \cite{kapinski2014simulation}. This is often done through a counterexample-guided iterative synthesis (CEGIS) procedure. The works in \cite{ravanbakhsh2015counter,ravanbakhsh2019learning} were among the first to employ the CEGIS procedure for Lyapunov function synthesis. The study in \cite{ahmed2020automated} also investigated the sound synthesis of Lyapunov functions and developed both LP- and SMT-based CEGIS implementations aimed at being simple and effective. One can also formulate the problem of identifying a polynomial Lyapunov function as a semi-algebraic problem and solve it using a computer algebra system (CAS). Notably, the works in \cite{she2013discovering,she2009semi} employ this approach to discover local Lyapunov functions and perform local stability analysis.

In this work, we aim to extend the use of SMT and LP solvers for the synthesis of global Lyapunov functions, addressing several gaps in the existing literature. First, for a candidate Lyapunov function to be a valid global Lyapunov function, both the Lyapunov function and its Lie derivative, i.e., the derivative along the vector field, must be globally sign-definite. For a polynomial vector field and a polynomial Lyapunov function candidate, simple examples reveal that such global sign-definiteness imposes specific restrictions on the Lie derivative and, as a result, on the Lyapunov function candidate, which can be checked and resolved algebraically. While one might expect a CEGIS procedure to identify these constraints via counterexamples, this is not always feasible if the template is initially unrestricted (see our motivating example in Section \ref{sec:necessary} and Claim \ref{claim} in Section \ref{sec:SR}). We take a first step in this direction by proposing an explicit symbolic reduction procedure that can be applied prior to optimization and SMT solving.

Second, not all globally asymptotically stable dynamical systems admit readily available strict Lyapunov functions, particularly if we restrict the search to polynomial functions of a given degree. A well-known result in dynamical systems, LaSalle's invariance principle \cite{lasalle1961stability,khalil2002nonlinear}, allows global asymptotic stability to be established using a weak Lyapunov function. However, applying LaSalle’s invariance principle often requires verifying that the largest invariant set within the set where the Lie derivative vanishes is, in fact, the origin, which is generally infeasible because it seemingly requires exact computation of a semi-algebraic set. We address this by developing an alternative sufficient condition that can be encoded as a satisfiability problem over polynomials.

Finally, stability analysis for polynomial vector fields remains an open question. Currently, it is unknown if the problem is decidable, with conjectures suggesting it is not \cite{arnold1976problems}. Additionally, it is unclear whether a locally asymptotically stable polynomial vector field, even a homogeneous polynomial one, necessarily admits a polynomial Lyapunov function \cite{ahmadi2013stability}. Practically, this implies that when a Lyapunov function of a specific form cannot be found, no definitive conclusion can be made about the system's stability. Developing Lyapunov certificates for instability would thus be beneficial in providing information in the opposite direction. While instability results are well known in dynamical systems stability analysis, they are not yet implemented within SMT frameworks.

To bridge these gaps, this paper develops automated symbolic reduction procedures that enable efficient and robust synthesis of global Lyapunov functions using optimization and SMT solvers. We also present a synthesis procedure that leverages LaSalle’s principle to prove global asymptotic stability using weak Lyapunov functions, as well as a method for conducting instability tests. Through a range of examples, from low- to high-dimensional systems, we demonstrate that symbolic reduction significantly improves synthesis efficiency and mitigates numerical challenges in both optimization- and SMT-based synthesis. The instability tests enable highly efficient proofs of instability for randomly generated dynamical systems with up to ten state dimensions.

\section{Preliminaries}\label{sec:prel}

Consider a nonlinear system described by 
\begin{equation}
    \label{eq:sys}
    \dot{x}  = f(x),
\end{equation}
where \( f:\, \Real^n \ra \Real^n \) defines the vector field, which is assumed to be a multivariate polynomial function in this paper. We also assume that $f(0)=0$. We denote the unique solution to (\ref{eq:sys}) from the initial condition $x(0)=x_0$ by $\phi(t,x_0)$ for $t$ in the maximal interval of existence for $\phi$ from $x_0$.  

We are concerned with the stability (or instability) of the equilibrium \( x=0 \) of system (\ref{eq:sys}). Of particular interest is the construction of Lyapunov functions that can either prove or disprove the global asymptotic stability of the origin. To this end, we provide the basic definition of global asymptotic stability and cite three well-known results from the literature on Lyapunov functions.

\begin{definition}
    The equilibrium $x=0$ of (\ref{eq:sys}) is said to be 
    \begin{itemize}
        \item stable if, for every $\eps>0$, there exists a $\delta>0$ such that $\abs{x_0}<\delta$ implies $\abs{\phi(t,x_0)}<\eps$ for all $t\ge 0$;
        \item asymptotically stable if it is stable and there exists some $\rho>0$ such that $\abs{x_0}<\rho$ implies $\phi(t,x_0)\ra 0$ as $t\ra \infty$;
        \item globally asymptotically stable if it is stable and $\phi(t,x_0)\ra 0$ as $t\ra \infty$ for all $x_0\in \Real^n$.
    \end{itemize}
\end{definition}

A function \( V: \mathbb{R}^n \to \mathbb{R} \) is said to be \textit{positive definite} if it satisfies \( V(0) = 0 \) and \( V(x) > 0 \) for all \( x \in \mathbb{R}^n \setminus \{0\} \). It is \textit{positive semi-definite} if \( V(0) = 0 \) and \( V(x) \ge 0 \) for all \( x \in \mathbb{R}^n \). We say \( V \) is \textit{negative definite} (or \textit{negative semi-definite}) if \( -V \) is positive definite (or positive semi-definite). Finally, \( V \) is \textit{radially unbounded} if \( V(x) \to \infty \) as \( |x| \to \infty \).

We start with the Lyapunov theorem for global stability. 

\begin{theorem}\cite[Theorem 4.2]{khalil2002nonlinear}
\label{thm:lyap}
    Suppose that there exists a continuously differentiable function $V:\,\Real^n\ra \Real$ that satisfies the following:
    \begin{enumerate}
        \item $V$ is positive definite and radially unbounded; and
        \item $\dot V=\nabla V\cdot f$ is negative definite. 
    \end{enumerate}
    Then $x=0$ is globally asymptotically stable (GAS) for (\ref{eq:sys}). 
\end{theorem}
With a slight abuse of terminology, we sometimes simply say (\ref{eq:sys}) is GAS if the origin is GAS for (\ref{eq:sys}). 

The previous theorem requires a strict Lyapunov function, in the sense that $\dot V$ is negative definite, to conclude global asymptotic stability. LaSalle's invariance principle provides an alternative to prove GAS with a weak Lyapunov function, where $\dot V$ is only required to be negative semi-definite. 

\begin{theorem}\cite[Corollary 4.2]{khalil2002nonlinear}\label{thm:lasalle}
        Suppose that there exists a continuously differentiable function $V:\,\Real^n\ra \Real$ satisfying:
    \begin{enumerate}
        \item $V$ is positive definite and radially unbounded;
        \item $\dot V=\nabla V\cdot f$ is negative semi-definite;
        \item the largest invariant set in $\set{x\in\Real^n:\, \dot V(x)=0}$ is $\set{0}$.
    \end{enumerate}
    Then $x=0$ is GAS for (\ref{eq:sys}). 
\end{theorem}

In fact, LaSalle's invariance principle \cite{lasalle1961stability} (see also \cite[Theorem 4.4]{khalil2002nonlinear}) is more general than the theorem cited above and states that any bounded solution of (\ref{eq:sys}) converges to the largest invariant set in $\set{x \in \Real^n : \nabla V(x) = 0}$, which implies Theorem \ref{thm:lasalle} as a special case. Furthermore, Theorem \ref{thm:lasalle} implies Theorem \ref{thm:lyap} as a special case where $\set{x \in \Real^n : \nabla V(x) = 0} = \set{0}$.

Finally, a Lyapunov function can also be used to infer instability. 

\begin{theorem}\cite[p.~38]{lasalle1961stability}
    \label{thm:instability}
    If there exists a continuously differentiable function $V:\,\Real^n\ra \Real$ such that $\dot{V}$ is negative semi-definite (or $\dot V(x)\le 0$ whenever $V(x)\le 0$) and some $z\in\Real^n$ such that $V(z)<V(0)$, then $x=0$ is not GAS for (\ref{eq:sys}). 
\end{theorem}

There are many variants of instability theorems \cite{lasalle1961stability}. We only discuss instability for GAS here. 
The proof of Theorem \ref{thm:instability} is straightforward. By the conditions in the theorem, we have \( V(\phi(t, z)) \le V(z) < V(0) \) for all \( t \ge 0 \), which prevents \( \phi(t, z) \rightarrow 0 \).

We adopt the multi-index notation to write polynomials and monomials. Let \( \alpha = (\alpha_1, \alpha_2, \dots, \alpha_n) \in \mathbb{N}^n \) be a multi-index representing the exponents of each variable. Here, \( \mathbb{N}^n \) denotes the set of all \( n \)-tuples of non-negative integers, allowing each \( \alpha_i \) to take any non-negative integer value. For example, we can write a polynomial $p$ of degree \( k \) without a constant term as 
\begin{equation}
\label{eq:poly}
p(x) = \sum_{1 \leq |\alpha| \leq k} c_\alpha x^\alpha = \sum_{1 \leq \alpha_1 + \alpha_2 + \dots + \alpha_n \leq k} c_\alpha x_1^{\alpha_1} x_2^{\alpha_2} \dots x_n^{\alpha_n}
\end{equation}
where \( |\alpha| = \alpha_1 + \alpha_2 + \dots + \alpha_n \) denotes the total degree of the monomial term \( x^\alpha \), and \( c_\alpha \) represents the corresponding coefficient.

\section{Necessary Conditions on Non-Positivity of Polynomials}\label{sec:main_necessary}

\subsection{A Motivating Example}

Searching for Lyapunov functions is usually done by fixing a candidate Lyapunov template and then determining parameters in this template to make it a valid Lyapunov function.

\begin{example}\label{ex:ex1}
    Consider
    \begin{equation}
       \dot x = f(x) = \begin{bmatrix}
            -x_1^3+x_1^5x_2\\
               -x_2^3-x_1^6
        \end{bmatrix}.
    \end{equation}
    We construct a quadratic Lyapunov function template of the form 
    $
    V(x) = c_0x_1^2 + c_1 x_2^2 + c_2x_1x_2 .
    $
    It can be easily computed that 
    \begin{align*}
          \dot V(x)  &  = \nabla V (x) \cdot f(x)  = 
   (2c_0 - 2c_1)x_1^6 x_2 - c_2 x_1^7 + c_2 x_1^5 x_2^2 \\
   & \qquad - 2c_0 x_1^4 - 2c_1 x_2^4  - c_2 x_1^3 x_2 - c_2 x_1 x_2^3. 
    \end{align*}
As we shall discuss in more detail in later sections, there are different approaches one can take to determine these coefficients through SMT solvers or optimization (e.g., gradient descent or LP). One of the necessary conditions we shall establish in Section \ref{sec:necessary} is that all highest and lowest odd monomials in $\dot V$ should vanish for it to be negative semi-definite. As a result, we can symbolically derive \( c_0 = c_1 \) and \( c_2 = 0 \), which means the only possible quadratic Lyapunov function for this system should be of the form  
\begin{equation}\label{eq:ex1_exact}
    V(x) = c_1(x_1^2 + x_2^2)
\end{equation}

On the one hand, if we use a complete SMT procedure to determine the coefficients such that \( V(x) > 0 \) and \( \dot V(x) \) hold for all \( x \in \Real^n \setminus \set{0} \), we should be able to recover this template (\ref{eq:ex1_exact}) and conclude that any \( c_1 > 0 \) provides a valid Lyapunov function. However, a complete SMT procedure can be slow for more than three decision variables, based on the evaluations we conduct. Therefore, being able to reduce the number of decision variables is crucial. 

On the other hand, if we were to use an optimization-based approach to solve for the parameters, any optimization procedure not uncovering (or retaining) the \textit{exact} template (\ref{eq:ex1_exact}) would not yield a valid Lyapunov function for this system. For example, one could potentially use gradient descent on the parameters with a loss function that encourages 
\begin{equation}\label{eq:lf_samples}
    V(y_j) > 0, \quad \dot V(y_j) < 0,
\end{equation}
for all \( y_j \) in a finite set of sample points \( \set{y_j}_{j=1}^N \). Alternatively, one can solve for \( (c_0, c_1, c_2) \) from (\ref{eq:lf_samples}) using LP, since these are linear constraints in the coefficients. Similarly, an incomplete SMT procedure can be used to solve for \( (c_0, c_1, c_2) \) such that (\ref{eq:lf_samples}) holds for a finite number of samples, potentially aided by a CEGIS procedure \cite{ahmed2020automated} that can add any counterexamples to the list of samples to improve the chance of discovering (\ref{eq:ex1_exact}).

Regardless of the procedure employed, as long as it is sample-based or requires numerical updates to the parameters \( (c_0, c_1, c_2) \), it will face the difficulty, if not impossibility, of discovering the exact form (\ref{eq:ex1_exact}). This limitation applies to several synthesis procedures we present in Section \ref{sec:synthesis} (see Claim \ref{claim}), and we will discuss it further in Section \ref{sec:examples} with numerical examples. This motivates the technical result in the next subsection, which can be used to symbolically reduce algebraic constraints such as \( c_2 = 0 \) and \( c_0 = c_1 \) in this example.
\end{example}

\subsection{Necessary Conditions on Non-Positivity of Polynomials}\label{sec:necessary}

We will focus on necessary conditions for the non-positivity of polynomials. 

\begin{proposition}\label{prop:sign-definite}
    Consider a polynomial $p$ of the form
    $$
    p(x) = \sum_{l \leq |\alpha| \leq k} c_\alpha x^\alpha, 
    $$
    where $l\ge 1$ and $k\ge l$ indicate the lowest and highest degrees of $p$, i.e., the lowest
    $l$ (or highest $k$) such that there exists some $\alpha\in \mathbb{N}^n$ with $\abs{\alpha}=l$ (or $k$) with $c_\alpha\neq 0$. Assume that $p(x)\le 0$ for all $x\in\Real^n$. Then the following statements are true:
    \begin{enumerate}
        \item The degrees $l$ and $k$ are even. Both \( p_l(x) = \sum_{|\alpha|=l} c_\alpha x^\alpha \) and \( p_k(x) = \sum_{|\alpha|=k} c_\alpha x^\alpha \) are negative semi-definite. 
        \item If any monomial \( c_\alpha x^\alpha \) in \( p(x) \) contains a factor \( x_i^{\alpha_i} \) whose degree is the highest among all monomials in \( p(x) \), then \( \alpha_i \) is even.
        \item If any monomial \( x^\alpha \) that appears in \( p_l(x) \) (or \( p_k(x) \)) contains a factor \( x_i^{\alpha_i} \) whose degree is the highest among all monomials in \( p_l(x) \) (or \( p_k(x) \)), then \( \alpha_i \) is even.
        \item If any monomial in $p$ takes the form $c x_i^d$ (i.e., with a single factor $x_i^d$ and coefficient $c$) and $d$ is either the lowest or highest degree among terms of the same form, then $d$ is even and $c\le 0$.
    \end{enumerate}
\end{proposition}

\begin{proof}
(1) Suppose that $k$ is odd. Pick $x$ such that \( p_k(x) = \sum_{|\alpha|=k} c_\alpha x^\alpha \) is nonzero. Consider any $\lambda\in\Real$. Then 
\[
p(\lambda x)=\lambda^k p_k(x) + \sum_{l \leq |\alpha| < k} c_\alpha (\lambda x)^\alpha.
\]
This is an odd-degree polynomial in $\lambda$ (note that the leading term coefficient $p_k(x)\neq 0$). As a result we cannot have $p(\lambda x)$ to be sign-definite for arbitrary $\lambda\in\Real$. Hence $k$ must be even. We also must have $p_k(x)\le 0$ in order for $p(\lambda x)\le 0$ for all $\lambda\in\Real$. Since $x$ is chosen arbitrarily, $p_k(x)$ must be negative semi-definite. 

Now suppose that $l$ is odd.
For $y\neq 0$, consider
\[
p\left(\frac{y}{\abs{y}^2}\right) = p_l\left(\frac{y}{\abs{y}^2}\right) + \sum_{l < |\alpha| \le k} c_\alpha \left(\frac{y}{\abs{y}^2}\right)^\alpha
\]
Multiply both sides by $\abs{y}^{2k}$. We obtain
\begin{align*}
   \hat p(y) & := \abs{y}^{2k} p\left(\frac{y}{\abs{y}^2}\right) = \underbrace{\abs{y}^{2k} p_l\left(\frac{y}{\abs{y}^2}\right)}_{:=\hat p_l(y)} + \underbrace{\abs{y}^{2k} \sum_{l < |\alpha| \le k} c_\alpha \left(\frac{y}{\abs{y}^2}\right)^\alpha}_{:=\hat R(y)}. 
\end{align*}
It follows that $\hat p_l(y)$ is a polynomial of degree $2k-l$ and $\hat R(y)$ is a polynomial of degree strictly less than $2k-l$. Note that $\hat p$ share the same sign with $p$. By the same argument above, where we showed the highest degree $k$ for $p$ cannot be odd, $2k-l$ cannot be odd because it is the highest degree for $\hat p$. Therefore, $l$ must be even and $\hat p_l(x)$ is negative semi-definite. Since $\hat p_l$ share the same sign with $p_l$, $p_l$ is also negative semi-definite. 

(2) Let $\set{c_\alpha x^{\alpha}}_{\alpha \in I}$ denote the set of monomials that contain a factor \( x_i^{\alpha_i} \) whose degree $\alpha_i$ is the highest among all monomials in \( p(x) \). We can write 
$
p(x) = \sum_{\alpha \in I} c_\alpha x^{\alpha} + \sum_{\alpha \not\in I} c_\alpha x^{\alpha}.
$
Let $x_\lambda=(x_1,\ldots,\lambda x_i,\ldots, x_n)$. Then 
$
p(x_\lambda) = \lambda^{\alpha_i} \sum_{\alpha \in I} c_\alpha x^{\alpha} + \sum_{\alpha \not\in I} c_\alpha (x_\lambda)^{\alpha}. 
$
Given any $x$ such that $\sum_{\alpha \in I} c_\alpha x^{\alpha}\neq 0$, $p(x_\lambda)$ can be viewed as a polynomial of $\lambda$, whose highest degree is $\alpha_i$. Since we must have $p(x_\lambda)\le 0$ for all $\lambda\in\Real$, $\alpha_i$ must be even. 

(3) Since \( p_k(x) \) and \( p_k(x) \) must be negative semi-definite as shown in part (1), we can apply the argument in part (2) to each of them to obtain the conclusion. 

(4) Pick \( x = (0, \ldots, x_i, \ldots, 0) \). We have \( p(x) = \sum_{l\le d\le k} c_d x_i^d \), which is a polynomial in \( x_i \). By part (1), the leading terms are negative semi-definite with respect to \( x_i \). This implies that $d$ is even and $c_d\le 0$ for both $d=l$ and $d=k$.
\end{proof}

\begin{rem}
    Clearly, if $p(x)\ge 0$ for all $x\in\Real^n$, we can draw the same conclusions as that of Proposition \ref{prop:sign-definite} with ``negative semi-definiteness'' replaced by ``positive semi-definiteness'' for $p_l(x)$ and $p_k(x)$ in part (1) and $c\ge 0$ replacing $c\le 0$ in part (4). 
\end{rem}

\begin{rem}
Determining the positivity of polynomials is an NP-hard problem. For this reason, the necessary conditions and the reduction rules derived from them are definitely not exhaustive. Furthermore, when combined with other incomplete synthesis procedures (e.g., linear programming or SMT solving based on inequality constraints evaluated at sample points), we can apply additional heuristic reductions to the parameters in the Lyapunov candidate. 
\end{rem}

\section{A Sufficient Condition for LaSalle's Invariance Principle}\label{sec:lasalle}

We use the following example to show that, although LaSalle's principle is a well-known result, a reformulation of its conditions may be necessary for it to be verifiable by an SMT solver in an automated manner.

\begin{example}
    Consider 
    \begin{equation}\label{eq:ex2_sys}
    \dot x = f(x) = \begin{bmatrix}
        x_2\\
        -x_1^3 - x_2^3
    \end{bmatrix}.
    \end{equation}
    Consider the Lyapunov function $V(x)=x_1^4+2x_2^2$. We can compute 
\begin{equation}
\label{eq:dot_V}
    \dot V(x) = \nabla V(x) \cdot f(x) = - 4x_2^4. 
\end{equation}
To apply LaSalle's invariance principle, we aim to show that the largest invariant set contained in \(\set{x\in\Real^2 :\, \dot V(x)=0} = \set{-4x_2^4 = 0}\) is \(\set{0}\). While it is straightforward to see in this particular example that the set \(\set{x \in \Real^2 :\, \dot V(x) = 0}\) is simply \(\set{x_2 = 0}\), we cannot generally expect to solve for this set symbolically. Therefore, we proceed with the analysis using the set \(\set{-4x_2^4 = 0}\). Differentiating $\dot V = -4x_2^4$ along solutions of (\ref{eq:ex2_sys}), we obtain 
\begin{equation}
\label{eq:ddot_V}
\ddot V(x) = \nabla (\dot V) \cdot f(x) =  -16 x_2^3 (-x_1^3 - x_2^3).     
\end{equation}
Intuitively, if 
\begin{equation}
\label{eq:one_step_lassale}
\dot V(x) = 0\quad \text{and} \quad x\neq 0 \Longrightarrow \ddot V(x)\neq 0,
\end{equation}
then the largest invariant set in $\set{\dot V = 0}$ must be $\set{0}$. This is because a solution $\phi(t,x)$ of (\ref{eq:ex2_sys}) with $\ddot V(x) \neq 0$ will necessarily leave the set $\set{\dot V = 0}$ instantly. In view of (\ref{eq:dot_V}) and (\ref{eq:ddot_V}), we do not have (\ref{eq:one_step_lassale}) hold. Nonetheless, we can continue with the same reasoning and seek to verify that there exists some integer $r \ge 1$ such that
\begin{equation}
\label{eq:multi_step_lassale}
L_f^k(x) = 0, \quad \forall k<r, \quad \text{and} \quad x\neq 0 \Longrightarrow L_f^r(x)\neq 0,
\end{equation}
where $L_f^k(x)$ denotes the $k$-th order derivatives of $V$ along solutions (formally defined in (\ref{eq:lie_derivative}) below). If this can be verified for some $r \ge 1$, then we know the largest invariant set in $\set{\dot V = 0}$ must be $\set{0}$ (see Proposition \ref{prop:lasalle} for a formal proof), and we can use LaSalle's invariance principle to conclude global asymptotic stability. For the example above, one can see that $L_f^5 V$ will expose a term that only involves $x_1^3$, while the rest of the terms all have a factor of $x_2$, which would imply (\ref{eq:multi_step_lassale}) holds. In fact, we can verify that $L_f^1(x) = \dot V(x) = 0$ and $x \neq 0$ imply $x_2 = 0$ and $x_1 \neq 0$, which in turn imply $L_f^5 V(x) \neq 0$. This provides the idea for the technical result below.
\end{example}

Consider system (\ref{eq:sys}) and a sufficiently smooth function $V:\, \Real^n \rightarrow \Real$. Define the various order Lie derivatives of $V$ with respect to (\ref{eq:sys}) as
\begin{equation}\label{eq:lie_derivative}
\begin{aligned}
    L_f^1 V(x) & = L_f^1 V(x) = \nabla V(x)\cdot f(x),\\
    L_f^{k+1} V(x) & = \nabla (L_f^k V(x) ) \cdot f(x),\quad k\ge 1.
\end{aligned}
\end{equation}
Introduce the sets
\begin{equation}\label{eq:C}
\begin{aligned}
    \mathcal C^1 & = \set{x\in\Real^n:\,L_f V(x)=0},\\
    \mathcal C^{k+1}  & = \set{x\in\Real^n:\,L_f^{k+1} V(x)=0} \cap \mathcal C^k,\quad k\ge 1.
\end{aligned}
\end{equation}

The following result provides a sufficient condition for applying LaSalle's invariance principle to prove global asymptotic stability. 

\begin{proposition}\label{prop:lasalle}
Suppose that there exists a sufficiently smooth, positive definite, and radially unbounded function $V:\, \Real^n \rightarrow \Real$ such that $\dot V(x)\le 0$ for all $x\in\Real^n$. Furthermore, there exists some integer $r\ge 1$ such that the sets defined in (\ref{eq:C}) satisfy
\begin{equation}\label{eq:Ck}
    \bigcap_{k=1}^r \mathcal C^{k} = \set{0},
\end{equation}
then the origin is globally asymptotically stable for (\ref{eq:sys}). 
\end{proposition}

\begin{proof}
   If $\mathcal{C}^1=\set{0}$, GAS follows from Theorem \ref{thm:lyap}. Now suppose that $\mathcal{C}^1\neq \set{0}$. 
   To apply LaSalle's invariance principle, we aim to show that, for any $x\in \mathcal C^1\setminus\set{0}$, the solution $\phi(t,x)$ cannot stay indefinitely in $\mathcal C^1$. Let $l\in \set{2,\ldots, r}$ be the smallest integer such that $x\not\in \mathcal C^l$. Such an $l$ exists because $x\neq 0$ and (\ref{eq:Ck}) holds. Then we have $L_f^{k}V(x)=0$ for $k=1,\ldots, l-1$ and $L_f^{l}V(x)\neq 0$. As a result, there exists a small $\tau_1>0$ such that $L_f^{l-1}V(\phi(t,x))\neq 0$ for all $t\in(0,\tau_1]$. Continuing this argument leads to the existence of a small $\tau_{l-1}>0$ such that $\dot V(\phi(t,x))\neq 0$ for all $t\in(0,\tau_{l-1}]$, which implies $\phi(t,x)$ must leave $\mathcal C^1$. Hence, the largest invariant set within the set $\mathcal C^1$ is $\set{0}$. By LaSalle's invariance principle (Theorem \ref{thm:lasalle}), the origin is GAS. 
\end{proof}

\begin{rem}
We note that similar conditions for computing the limit invariant set and synthesizing relaxed (weak) Lyapunov functions using LaSalle's principle have been derived in \cite{liu2012automatically,gerbet2020proving}. In particular, \cite{gerbet2020proving} (Proposition 1 and the remarks following Proposition 3) notes that the finite termination of a condition analogous to (\ref{eq:Ck}) can be established using the descending chain condition for algebraic varieties, as discussed in \cite{cox2015ideals}.
\end{rem}

\section{Synthesis Algorithms}\label{sec:synthesis}

In this section, we discuss several algorithmic procedures that can be used to efficiently synthesize a polynomial Lyapunov function with the aid of SMT and optimization solvers. 

\subsection{Symbolic Reduction}\label{sec:reduction}

Based on Proposition \ref{prop:sign-definite}, we can formulate the following symbolic reduction procedure. Consider a Lyapunov function template 
\begin{equation}
\label{eq:lyap}
V(x) = \sum_{\alpha\in I} c_\alpha x^\alpha. 
\end{equation}
The set of multi-indices, $I$, specifies which monomial terms are included in the Lyapunov function candidate. For example, one may wish to include all monomial terms in a quadratic form. For higher degrees, it is common practice to retain only even-degree monomials in \( V \). For a polynomial vector field $f$, we then compute the Lie derivative as
\begin{equation}
\label{eq:lyap2}
\dot V(x) = \nabla V(x)\cdot f(x) = \sum_{\alpha \in \hat I} \hat c_\alpha x^\alpha. 
\end{equation}
The possible terms that appear in $\hat I$ are obviously determined by that of $f$ and $V$. For $V$ to be a valid global Lyapunov function candidate (for both Lyapunov and LaSalle conditions), we require at least $\dot V(x)\le 0$ for all $x\in\Real^n$. Therefore, according to Proposition \ref{prop:sign-definite}, certain terms cannot appear. We can use the result of Proposition \ref{prop:sign-definite} to introduce symbolic equality and inequality constraints on $\set{\hat c_\alpha}_{\alpha \in}$ and solve them to reduce the number of terms in $V$. 

We revisit Example \ref{ex:ex1} to illustrate the procedure.

\begin{example}[Example \ref{ex:ex1} revisited]\label{ex:ex1:ex1}
The Lyapunov candidate is  
$$
V(x) = c_0x_1^2 + c_1 x_2^2 + c_2x_1x_2.
$$
and the computed Lie derivative \(\dot{V}\) is 
$$
(2c_0 - 2c_1)x_1^6 x_2 - c_2 x_1^7 + c_2 x_1^5 x_2^2 - 2c_0 x_1^4 - 2c_1 x_2^4  - c_2 x_1^3 x_2 - c_2 x_1 x_2^3. 
$$
By Proposition \ref{prop:sign-definite}(1), for \(\dot{V}\) to be negative semi-definite, the highest-degree terms must be even. Setting \(2c_0 - 2c_1=0\), \(- c_2=0\), and \(c_2=0\) yields \(c_0=c_1\) and \(c_2=0\). This results in the only valid quadratic Lyapunov function candidate, \(V(x)=c_1(x_1^2 + x_2^2)\).
\end{example}

\subsection{Complete Synthesis}

Given a Lyapunov function candidate $V(x) = \sum_{\alpha\in I} c_\alpha x^\alpha 
$, possibly after symbolic reduction outlined in Section \ref{sec:reduction}, we propose the following SMT procedures to synthesize global Lyapunov functions for GAS either using a strict Lyapunov function or a weak Lyapunov function by LaSalle's invariance principle. We also discuss how to synthesize a Lyapunov function to prove instability. The procedure is said to be complete in the sense that if the formulated SMT conditions encoding the Lyapunov conditions on a particular Lyapunov function template are satisfiable, then a Lyapunov function can be found. If they are not satisfiable, then a Lyapunov function of this particular form does not exist to satisfy the given SMT conditions. 

\subsubsection{Strict Lyapunov function:} 

The satisfiability problem for finding a strict Lyapunov function can be formulated as follows:
\begin{equation}
    \label{eq:smt_complete}
\exists \{c_\alpha\} ( (\forall x (x \neq 0 \Rightarrow  (V(x) > 0 \wedge \nabla V(x) \cdot f(x) < 0)))).
\end{equation}
For a polynomial Lyapunov function \( V \) and a polynomial vector field \( f \), the SMT solver Z3 \cite{de2008z3} can readily solve this problem. If the satisfiability check is successful, a ``model'' — an assignment of values to the coefficients \( \{c_i\} \) that satisfies all constraints — can be extracted, defining a valid Lyapunov function. By Theorem \ref{thm:lyap}, this proves GAS of (\ref{eq:sys}), provided that $V$ is also radially unbounded. 

\subsubsection{Weak Lyapunov function and GAS using LaSalle's principle:}

If a strict Lyapunov function cannot be found, GAS can still be established using LaSalle's invariance principle. In this case, the SMT procedure relaxes the conditions to search for a weak Lyapunov function \( V \) where:
\begin{equation}
    \label{eq:smt_weak}
\exists \{c_\alpha\} ( (\forall x (x\neq 0 \Rightarrow (V(x) > 0 \wedge \nabla V(x) \cdot f(x) \leq 0)))).
\end{equation}
Additionally, we use Proposition \ref{prop:lasalle} to formulate the following LaSalle's condition:
\begin{equation}
    \label{eq:lasalle}
\forall x ( (\nabla V \cdot f(x) = 0 \wedge x\neq 0) \Rightarrow \vee_{2\le k\le r} (L_f^k (x) \neq 0)).
\end{equation}
Here, $r$ defines the number of Lie derivatives we are checking LaSalle's condition with. Note that (\ref{eq:lasalle}) is equivalent to (\ref{eq:Ck}). A simpler sufficient condition for (\ref{eq:lasalle}) is 
\begin{equation}
    \label{eq:lasalle2}
\forall x ( (\nabla V \cdot f(x) = 0 \wedge x\neq 0) \Rightarrow  L_f^r (x) \neq 0),
\end{equation}
for some $r\ge 2$. Both (\ref{eq:lasalle}) and (\ref{eq:lasalle2}) can be easily coded in an iterative loop to check if a Lyapunov function of a particular form can be found using LaSalle's condition up to a degree-$r$ Lie derivative. 

\begin{rem}
There is a choice to be made regarding whether to code the LaSalle conditions directly in the synthesis part, or to first synthesize a weak Lyapunov function satisfying (\ref{eq:smt_weak}), and then verify the LaSalle condition (\ref{eq:lasalle}) or (\ref{eq:lasalle2}). Verification is presumably cheaper than synthesis, but synthesis can be comprehensive in determining whether a weak Lyapunov function of a particular exists that satisfies a LaSalle-type condition, whereas a weak Lyapunov function synthesized without the LaSalle condition may not always verify the LaSalle condition. 
\end{rem}

\subsubsection{Lyapunov function for instability}

According to Theorem \ref{thm:instability}, if the following condition: 
\begin{equation*}
    \exists \{c_\alpha\} (( \forall x ( V(x) \leq 0 \Rightarrow \nabla V(x) \cdot f(x) \leq 0 )\wedge \exists z (V(z) < V(0))))
\end{equation*}
is satisfiable, then we can use $V$ to prove that the origin is not globally asymptotically stable for (\ref{eq:sys}). Even though the condition does not fully require \(\dot{V}\) to be negative semi-definite, we can still use the symbolic procedure outlined in Section \ref{sec:reduction} to potentially reduce the number of decision variables in SMT synthesis. This turns out to be highly effective, and we demonstrate in Section \ref{sec:examples} that we can efficiently disprove GAS of non-trivial examples up to 10 dimensions in under a second for each example.

\subsection{Sampling-Based Partial Synthesis}

In contrast to the complete synthesis method discussed in the previous section, one can also use a sampling-based partial synthesis approach to compute a candidate Lyapunov function and then formally verify it with respect to the Lyapunov conditions. This approach is especially useful when exact SMT synthesis may be computationally prohibitive. Sampling-based approaches are usually combined with a counterexample-guided iterative synthesis (CEGIS) approach \cite{ravanbakhsh2019learning,ahmed2020automated} to synthesize valid Lyapunov functions. 

\subsubsection{LP-CEGIS}

The Lyapunov condition (\ref{eq:smt_complete}) depends linearly on the coefficients $\set{c_\alpha}$. Instead of solving $\set{c_\alpha}$ such that (\ref{eq:smt_complete}) holds for all $x\in\Real^n$, we can sample a number of points $\set{y_j}_1^N\in\Real^n\setminus\set{0}$ and use them formulate $2N$ linear constraints in $\set{c_\alpha}$:
\begin{equation}
    \label{eq:LP_constraints}
    \begin{aligned}
\sum_{\alpha \in I} c_\alpha (y_j)^\alpha > 0 ,\quad j=1,\ldots, N, \\
\sum_{i=1}^n \sum_{\alpha \in I} c_\alpha \alpha_i y_j^{\alpha - e_i} f_i(y_j) < 0,\quad j=1,\ldots, N,
    \end{aligned}
\end{equation}
where $e_i = (0, 0, \ldots, 1, \ldots, 0)$. The purpose of writing this down is simply to illustrate that these constraints are indeed linear in \(\{c_\alpha\}\). In practice, the evaluation and formulation of these constraints can be automated. However, solving (\ref{eq:LP_constraints}) directly with linear programming (LP) will very likely return \(c_\alpha = 0\) for all \(\alpha \in I\). To obtain nontrivial solutions, we need to reformulate the problem. We implemented the following as a benchmark for solving the examples presented in this paper: 
\begin{equation}
    \label{eq:LP_constraints_2}
    \begin{aligned}
\sum_{\alpha \in I} c_\alpha (y_j)^\alpha \ge \mu \min(|y_j|^{l_V}, |y_j|^{k_V}) ,\quad j=1,\ldots, N, \\
\sum_{i=1}^n \sum_{\alpha \in I} c_\alpha \alpha_i y_j^{\alpha - e_i} f_i(y_j) \le - \mu\min(|y_j|^{l_{\dot V}}, |y_j|^{k_{\dot V}}),\quad j=1,\ldots, N,
    \end{aligned}
\end{equation}
where $\mu>0$ is a small constant, \( l_V \) and \( l_{\dot V} \) are the lowest degrees of monomials in \( V \) and \( \dot V \), respectively, and \( k_V \) and \( k_{\dot V} \) are the highest degrees. We can solve (\ref{eq:LP_constraints_2}) to obtain a Lyapunov candidate. While this may not yield a valid Lyapunov function, one can use an SMT solver to further verify the strict Lyapunov condition (\ref{eq:smt_complete}) or the weak Lyapunov condition (\ref{eq:smt_weak}), along with the LaSalle condition (\ref{eq:lasalle}) and or (\ref{eq:lasalle2}). If a counterexample is found during verification of \( V \), this counterexample can be added to the set of samples to re-solve for \( \{c_\alpha\} \). This procedure is referred to as LP-CEGIS.

\subsubsection{Z3-CEGIS}

Instead of using an LP solver, one can also directly use an SMT solver to solve for \(\{c_\alpha\}\) from the constraints (\ref{eq:LP_constraints}) provided by a finite set of samples, although this approach is potentially more expensive than solving an LP. When the SMT solver is Z3, we refer to this as Z3-CEGIS. Notably, both LP-CEGIS and Z3-CEGIS were previously proposed in \cite{ahmed2020automated} for the synthesis of Lyapunov functions. We implement these as baselines for the numerical case studies in this paper.

\subsection{LP-CEGIS-SR and Z3-CEGIS-SR}\label{sec:SR}

We revisit Example \ref{ex:ex1} (and Example \ref{ex:ex1:ex1}) to illustrate the necessity of using symbolic reduction when applying a CEGIS procedure. We prove the following claim.

\begin{claim}\label{claim}
Consider Example \ref{ex:ex1} and a quadratic Lyapunov function candidate $
V(x) = c_0x_1^2 + c_1 x_2^2 + c_2x_1x_2
$. For any finite set of samples $\set{y_j}_1^{N}$, there exist $\set{c_0,c_1,c_2}$ such that 
$c_0\neq c_1$ and $V$ satisfies both (\ref{eq:LP_constraints}) and (\ref{eq:LP_constraints_2}). 
\end{claim}
\begin{proof}
Recall that the Lie derivative of $V$ is 
$$
(2c_0 - 2c_1)x_1^6 x_2 - c_2 x_1^7 + c_2 x_1^5 x_2^2 - 2c_0 x_1^4 - 2c_1 x_2^4  - c_2 x_1^3 x_2 - c_2 x_1 x_2^3. 
$$
We can simply let $c_2=0$. The positive definiteness constraints are always satisfied, provided that $c_0$ and $c_1$ are sufficiently large positive values relative to $\mu$. We have  
$$
\dot V(x) = (2c_0 - 2c_1)x_1^6 x_2 - 2c_0 x_1^4 - 2c_1 x_2^4. 
$$
For any sample \( y_j = (y_{j1}, y_{j2}) = 0 \), the constraints (\ref{eq:LP_constraints}) and (\ref{eq:LP_constraints_2}) are trivially satisfied. Hence, assume, without loss of generality, that \( y_j \neq 0 \) for all \( j = 1, \ldots, N \). For sufficiently large \( c_0 \) and \( c_1 \), each constraint from (\ref{eq:LP_constraints}) and (\ref{eq:LP_constraints_2}) takes the form \( (2c_0 - 2c_1) b_j \le d_j \), where \( d_j > 0 \). As a result, we can choose \( c_0 \) to be sufficiently close, but not equal, to \( c_1 \) to satisfy all the constraints.
\end{proof}

By Proposition \ref{prop:sign-definite} in Section \ref{sec:reduction}, any \( c_0 \neq c_1 \) makes \( V \) an invalid Lyapunov function. Adding further counterexamples will not help either, due to the claim above. This simple analysis reveals that symbolic reduction, as discussed in Section \ref{sec:reduction}, can play a crucial role in identifying and eliminating algebraic constraints before optimization and satisfiability solvers are effectively deployed for Lyapunov function synthesis.
We observe in our implementations that CEGIS (without symbolic reduction and with exact arithmetic) fails to converge on this example. 

\begin{rem}\label{rem:rounding}
We note, however, that rounding the solutions returned by LP or SMT solvers can mitigate the issue raised in Claim \ref{claim} because, in the analysis above, we would need to choose \( c_0 \) and \( c_1 \) sufficiently close with increasing number of samples. Nonetheless, this approach would be ad hoc and akin to how we filter out small coefficients returned by SOS programming. Hence, we believe that symbolic reduction is a more principled way of handling these constraints. Furthermore, an obvious benefit of using symbolic reduction is to reduce the number of decision variables in optimization or SMT synthesis. 
\end{rem}

\begin{rem}
Intuitively, symbolic reduction aims (albeit without guarantees) to reduce \( V \) to the form \( V(x) = \sum_{i=1}^N c_i p_i(x) \), where each \( p_i(x) \) is a polynomial, such that the set  
$$
\Theta = \{c = (c_1, \ldots, c_N) : V \text{ satisfies (\ref{eq:smt_weak})}\}
$$  
has a positive measure in \( \mathbb{R}^N \). If this condition holds, it is expected that computing \( c \) can be achieved more robustly.
\end{rem}

In the next section, we refer to LP-CEGIS and Z3-CEGIS with symbolic reduction as \textbf{LP-CEGIS-SR} and \textbf{Z3-CEGIS-SR}, respectively. We refer to complete synthesis with symbolic reduction as Z3-Complete-SR. We encode the LaSalle condition (\ref{eq:lasalle2}) for synthesis in Z3-Complete-SR-LaSalle and for verification of weak Lyapunov functions in \textbf{LP-CEGIS-SR-LaSalle} and \textbf{Z3-CEGIS-SR-LaSalle}.

\begin{rem}
We omit the discussion of training polynomial Lyapunov functions using gradient descent in this paper. To do so, one can use a quadratic activation function in a multi-layer neural network and train it with a loss function that encourages the inequality constraints (\ref{eq:LP_constraints}) or (\ref{eq:LP_constraints_2}) to be satisfied; see, e.g., \cite{edwards2024fossil,liu2024tool}. A typical gradient descent (GD) optimization algorithm can be called to minimize the loss. This approach can also be combined with a CEGIS procedure as implemented in \cite{edwards2024fossil}. Such an approach can be called GD-CEGIS or NN-CEGIS (NN stands for neural networks). We implemented this, but the results are not comparable with LP-CEGIS, and it suffers from the same drawbacks as LP-CEGIS and Z3-CEGIS, where counterexamples cannot be exhausted due to certain underlying algebraic constraints (stipulated by Proposition \ref{prop:sign-definite} and illustrated by Claim \ref{claim}) that are not eliminated before optimization/SMT solvers are deployed.
\end{rem}

\section{Case Studies and Experiments}\label{sec:examples}

In this section, we revisit a suite of examples from related work \cite{papachristodoulou2002construction,ahmed2020automated,sankaranarayanan2013lyapunov} and introduce new examples to illustrate the main results of this paper. In particular, we contrast different approaches: Z3-Complete-SR, LP-CEGIS-SR, Z3-CEGIS-SR, and the baseline LP-CEGIS and Z3-CEGIS without symbolic reduction. We also include results obtained using SOSTOOLS \cite{papachristodoulou2013sostools} for comparison. All computations were performed on an Intel Xeon Gold 6326 CPU @ 2.90 GHz with 32 cores and 16 GB of RAM on a computing cluster. Implementations are done using the LyZNet toolbox \cite{liu2024tool}.

\paragraph{Data Availability.} The models, scripts, and tools to reproduce our experimental evaluation are archived and publicly available at DOI~\href{https://doi.org/10.5281/zenodo.15272621}{10.5281/zenodo.15272621}.

A more detailed description of the examples, along with the Lyapunov functions successfully synthesized using different approaches, is provided in the Appendix. 

We run experiments on a total of 10 examples, which are polynomial vector fields of dimensions 2 to 6, labeled as (E1)–(E10) and listed below:
\[
(\textbf{E1})\, \dot{x} = \begin{bmatrix} -x_1^3 + x_1^5 x_2 \\ -x_2^3 - x_1^6 \end{bmatrix},\;
(\textbf{E2})\, \dot{x} = \begin{bmatrix} -x_1^7 + x_1 x_2 \\ -x_2^7 - x_1^2 \end{bmatrix},\;
(\textbf{E3})\, \dot{x} = \begin{bmatrix} -x_1 - 1.5 x_1^2 x_2^3 \\ -x_2^3 + 0.5 x_1^3 x_2^2 \end{bmatrix},
\]
\[
(\textbf{E4})\, \dot{x} = \begin{bmatrix} -x_1^3 + x_2 \\ -x_1 - x_2 \end{bmatrix},\;
(\textbf{E5})\, \dot{x} = \begin{bmatrix} -\sigma x_1 + \sigma x_2 \\ r x_1 - x_2 - x_1 x_3 \\ -b x_3 + x_1 x_2 \end{bmatrix},\,\sigma=10,\,r=0.9999,\,b=\frac83,
\]
\[
(\textbf{E6})\, \dot{x} = \begin{bmatrix} -x_1 + x_2^3 - 3 x_3 x_4 \\ -x_1 - x_2^3 \\ x_1 x_4 - x_3 \\ x_1 x_3 - x_4^3 \end{bmatrix},\;
(\textbf{E7})\, \dot{x} = \begin{bmatrix} 
-x_1^3 + 4 x_2^3 - 6 x_3 x_4 \\ 
-x_1 - x_2 + x_5^3 \\ 
x_1 x_4 - x_3 + x_4 x_6 \\ 
x_1 x_3 + x_3 x_6 - x_4^3 \\ 
-2 x_2^3 - x_5 + x_6 \\ 
-3 x_3 x_4 - x_5^3 - x_6 
\end{bmatrix}.
\]
\[
(\textbf{E8})\, \dot{x} = \begin{bmatrix} x_2 \\ -x_1^3 - x_2^3 \end{bmatrix},
\;
(\textbf{E9})\, \dot{x} = \begin{bmatrix} x_2 \\ -x_1^5 - 3 x_2 \end{bmatrix},
\;
(\textbf{E10})\, \dot{x} = \begin{bmatrix} x_2 \\ -x_1 - 7 x_2^5 \end{bmatrix}.
\]
The main objective is to assess whether symbolic reduction can improve the likelihood of discovering a valid Lyapunov function. Results, including computational times, for (E1)–(E7) are reported in Table~\ref{tab:stability}, and for (E8)--(E10) in Table~\ref{tab:lasalle}. 

For sampling-based approaches, the initial domain for all examples was set to $[-10,10]^n$, but subsequent samples (counterexamples) are drawn from $\mathbb{R}^n$. We selected $3{,}000$ initial samples for LP-CEGIS/LP-CEGIS-SR and $300$ for Z3-CEGIS/Z3-CEGIS-SR, as the latter is considerably slower and significantly impacted by the number of constraints and decision variables. We do a total of 10 CEGIS steps for all iterative synthesis approaches.  While LP\footnote{For the examples in this paper, we solved LP using \texttt{scipy.optimize.linprog} \cite{virtanen2020scipy} with the HiGHS solver \cite{huangfu2023highs}.} 
 is solved with floating-point arithmetic, we convert it\footnote{We use \texttt{symp.Rational} on the solution returned by LP, rounding it to a specified precision, and apply the same process to the solution returned by Z3. We observed that rounding has two benefits: first, it mitigates the issue of unsolved constraints (see Remark \ref{rem:rounding}), and second, it provides a more interpretable expression of the Lyapunov functions.} to exact rational arithmetic when extracting the Lyapunov function expression for verification. For solving with SOSTOOLS \cite{papachristodoulou2013sostools}, we use the build-in function $\texttt{findlyap}$ with the default solver, SeDuMi \cite{sturm1999using}.

Results in Table~\ref{tab:stability} show that incorporating symbolic reduction clearly increases the likelihood of discovering a valid Lyapunov function in most cases. Furthermore, Table~\ref{tab:lasalle} shows that with symbolic reduction, all approaches using LaSalle's principle successfully find a Lyapunov function, whereas approaches without LaSalle fail to do so. Finally, we also conducted extensive tests to evaluate the effect of symbolic reduction on instability detection using Z3. As shown in Table \ref{tab:instability}, symbolic reduction significantly improves performance, enabling successful verification even for high-dimensional systems, while the baseline approach without reduction fails (100\% timeout) even in low dimensions.

\begin{table*}[h!]
\caption{Runtimes (in seconds) of Lyapunov function synthesis methods for the systems in Examples (E1)--(E7). Symbolic reduction clearly plays a crucial role in increasing the likelihood of discovering a valid Lyapunov function in most cases.}
\label{tab:stability}
\centering
\begin{tabular}{|l|c|c|c|c|c|c|c|}
\hline
 & Dim & Z3-Complete-SR & LP-CEGIS & LP-CEGIS-SR & Z3-CEGIS & Z3-CEGIS-SR & SOS\\ \hline
E1 & 2 & 0.226 & 0.193 & 0.075 & 4.863 & 1.402 & 3.200 \\
E2 & 2 & 0.202 & 0.168 & 0.071 & - & 1.479 & 1.566 \\
E3 & 2 & 0.245 & 0.193 & 0.101 & - & 1.483 & - \\
E4 & 2 & 0.206 & 0.167 & 0.077 & - & 1.429 & - \\
E5 & 3 & 0.357 & - & 0.170 & - & 19.090 & 1.629\\
E6 & 4 & - & - & 0.113 & - & - & - \\
E7 & 6 & - & - & 0.349 & - & - & - \\
\hline
\end{tabular}
\end{table*}

\begin{table*}[h!]
\caption{Runtimes (in seconds) of Lyapunov and LaSalle synthesis methods for the three systems in (E8)--(E10). Synthesis without the LaSalle module failed to find a valid Lyapunov function in all cases.}
\label{tab:lasalle}
\centering
\begin{tabular}{|l|c|c|c|c|c|c|c|}
\hline
& Dim & Z3-Complete-SR-LaSalle & LP-CEGIS-SR-LaSalle & Z3-CEGIS-SR-LaSalle \\ \hline
E8 & 2 & 0.412 & 0.291 & 1.449 \\
E9 & 2 &  0.097 & 0.135 & 1.382 \\
E10 & 2 & 1.032 & 0.609 & 1.771 \\
\hline
\end{tabular}
\end{table*}

\begin{table*}[h!]
\caption{Experiment results for ``speedy'' instability tests using the SMT solver Z3: A total of 100 systems is randomly generated for each case, with a prescribed maximum degree for the monomials. The number of monomial terms for each dimension is set to 3. Systems that are trivially unstable, where the linearization has an eigenvalue with a positive real part, are eliminated from the trials. We also exclude vector fields that contain a zero component, as they cannot be asymptotically stable. A timeout (t.o.) of 1 second is applied to each call of Z3-complete-SR for stability or instability tests. The table includes the average synthesis time for stable and unstable cases, encompassing both SMT solver time and the overhead for symbolic reduction. Notably, a considerable number of systems can be efficiently proven to be not globally asymptotically stable in under 1 second, even for systems up to dimension 10.}

\begin{tabular}{lllllllll}
\toprule
\multicolumn{2}{c|}{} & \multicolumn{3}{c|}{Unstable} & \multicolumn{3}{c|}{Stable} & \\
\cmidrule(r){3-5} \cmidrule(r){6-8}
Dim ($n$) & Deg & Unstable \% & t.o. (\%) & Time (s) & Stable \% & t.o. (\%) & Time (s) & No. Systems \\
\midrule
2 & 2 & 61.0 & 1.0 & 0.033 & 3.0 & 0.0 & 0.045 & 100 \\
2 & 3 & 34.0 & 7.0 & 0.033 & 1.0 & 21.0 & 0.050 & 100 \\
3 & 2 & 57.0 & 4.0 & 0.038 & 0.0 & 3.0 & 0.000 & 100 \\
3 & 3 & 38.0 & 20.0 & 0.068 & 0.0 & 48.0 & 0.000 & 100 \\
4 & 2 & 58.0 & 8.0 & 0.061 & 0.0 & 6.0 & 0.000 & 100 \\
4 & 3 & 25.0 & 21.0 & 0.075 & 0.0 & 60.0 & 0.000 & 100 \\
5 & 2 & 53.0 & 8.0 & 0.066 & 0.0 & 6.0 & 0.000 & 100 \\
5 & 3 & 27.0 & 18.0 & 0.102 & 0.0 & 58.0 & 0.000 & 100 \\
6 & 2 & 58.0 & 2.0 & 0.103 & 0.0 & 1.0 & 0.000 & 100 \\
6 & 3 & 23.0 & 24.0 & 0.158 & 0.0 & 65.0 & 0.000 & 100 \\
7 & 2 & 50.0 & 4.0 & 0.174 & 0.0 & 3.0 & 0.000 & 100 \\
7 & 3 & 23.0 & 24.0 & 0.272 & 0.0 & 70.0 & 0.000 & 100 \\
8 & 2 & 58.0 & 4.0 & 0.305 & 0.0 & 0.0 & 0.000 & 100 \\
8 & 3 & 25.0 & 9.0 & 0.379 & 0.0 & 53.0 & 0.000 & 100 \\
9 & 2 & 60.0 & 1.0 & 0.443 & 0.0 & 0.0 & 0.000 & 100 \\
9 & 3 & 24.0 & 24.0 & 0.524 & 0.0 & 61.0 & 0.000 & 100 \\
10 & 2 & 51.0 & 2.0 & 0.608 & 0.0 & 2.0 & 0.000 & 100 \\
10 & 3 & 12.0 & 21.0 & 0.718 & 0.0 & 90.0 & 0.000 & 100 \\
\bottomrule
\end{tabular}

\label{tab:instability}
\end{table*}

\section{Conclusions}\label{sec:conclusion}

In this paper, we investigated automated synthesis of global Lyapunov functions and proposed a simple set of algebraic checks that can be used to symbolically solve for parameters in a polynomial Lyapunov function candidate, before an optimization or SMT-based synthesis procedure is used to compute a global Lyapunov function. We also proposed sufficient LaSalle-type conditions that can be implemented in an SMT procedure, either as verification or synthesis, for global Lyapunov functions using a weak Lyapunov function candidate. Additionally, we encoded instability conditions and designed SMT-based procedures for disproving global asymptotic stability. Through a suite of examples, we demonstrated that the proposed symbolic reduction, LaSalle-type conditions, and instability tests allow us to efficiently solve many examples that would be otherwise unsolvable.

For future work, one could further build upon the symbolic reduction rules and encode additional algebraic necessary conditions for sign-definite polynomials. It would be interesting to investigate how the symbolic reduction procedure could be developed for non-polynomial vector fields. We expect that if the vector field is locally smooth, parts of the reduction rules may still apply. Similarly, if the vector field behaves well at infinity, similar reductions could be performed. While this paper focused on global asymptotic stability, one could also specialize to local stability, for which symbolic reduction, LaSalle's principle, and instability tests could be developed accordingly.

\newpage

\section*{Acknowledgements}

This work was supported in part by the Natural Sciences and Engineering Research Council of Canada and the Canada Research Chairs Program. Jun Liu would like to thank Alessandro Abate, Alec Edwards, and Andrea Peruffo for helpful discussions.

\bibliographystyle{splncs04}
\bibliography{refs}


\appendix

\section{Description of Examples in Case Studies}\label{app:examples}

In this section, we provide details of the examples (E1)--(E10) included in the case studies in Section \ref{sec:examples}. We also include the Lyapunov functions found using different approaches.

\textbf{(E1)} This is the model in Example \ref{ex:ex1}, i.e.,
    \begin{equation}
       \dot x = f(x) = \begin{bmatrix}
            -x_1^3+x_1^5x_2\\
               -x_2^3-x_1^6
        \end{bmatrix}.
    \end{equation}
The following Lyapunov functions are returned by different synthesis approaches:
    \begin{itemize}
        \item \textbf{Z3-Complete-SR}: \( V(x) = x_1^2 + x_2^2 \).
        \item \textbf{LP-CEGIS}: \( V(x) = 197 x_1^2/100 + 197 x_2^2/100 \).
        \item \textbf{LP-CEGIS-SR}: \( V(x) = x_1^2/100 + x_2^2/100 \). 
        \item \textbf{Z3-CEGIS} \( V(x) = 47 x_1^2/50 + 47x_2^2/50 \).
        \item \textbf{Z3-CEGIS-SR}: \( V(x) =  x_1^2 + x_2^2 \). 
    \end{itemize}
SOSTOOLS provides $V = 1.179 x_1^2 - 3.34\times 10^{-22}x_1x_2 + 1.179x_2^2$, which is a valid Lyapunov function after omitting the small term.

\textbf{(E2)} Consider the model
\[
\dot{x} = \begin{bmatrix} -x_1^7 + x_1 x_2 \\ -x_2^7 - x_1^2 \end{bmatrix}.
\]
The synthesis results are as follows: 
\begin{itemize}
    \item \textbf{Z3-Complete-SR}: \( V(x) = x_1^2 + x_2^2 \).
    \item \textbf{LP-CEGIS}: \( V(x) = x_1^2/25 + x_2^2/25 \).
    \item \textbf{LP-CEGIS-SR}: \( V(x) = x_1^2/25 + x_2^2/25 \).
    \item \textbf{Z3-CEGIS} fails to produce a valid Lyapunov function after 10 CEGIS steps.    
    \item \textbf{Z3-CEGIS-SR}: \( V(x) = x_1^2 + x_2^2 \). 
\end{itemize}
SOSTOOLS finds $V(x) = 1.637 x_1^2 - 2.097 \times 10^{-16}x_1x_2 + 1.637x_2^2$, which is a valid Lyapunov function after omitting the negligible term. 

\textbf{(E3)}
    Consider the model \cite{sankaranarayanan2013lyapunov}
\[
\dot{x} = \begin{bmatrix} -x_1 - 1.5 x_1^2 x_2^3 \\ -x_2^3 + 0.5 x_1^3 x_2^2 \end{bmatrix}.
\]
The computational results are as follows:
\begin{itemize}
    \item \textbf{Z3-Complete-SR}: \( V(x) = x_1^2/3 + x_2^2 \).
    \item \textbf{LP-CEGIS}: \( V(x) = x_1^2/20 + 3 x_2^2/20 \).
    \item \textbf{LP-CEGIS-SR}: \( V(x) = 37 x_1^2/75 + 37 x_2^2/25 \).
    \item \textbf{Z3-CEGIS} fails to produce a valid Lyapunov function after 10 CEGIS steps.    
    \item \textbf{Z3-CEGIS-SR}: \( V(x) = x_1^2/3 + x_2^2 \). 
\end{itemize}
SOSTOOLS returns \( V(x) = 0.7424x_1^2 + 1.406\times 10^{-9}x_1x_2 + 2.227x_2^2 \) and the work in \cite{sankaranarayanan2013lyapunov} gives \( V(x) = 0.2x_1^2 + x_2^2 \). Neither can be verified as a global Lyapunov function (even after removing the negligible term from the one returned by SOSTOOLS).

\textbf{(E4)} 
    Consider the model \cite{sankaranarayanan2013lyapunov}
\[
\dot{x} = \begin{bmatrix} -x_1^3 + x_2 \\ -x_1 - x_2 \end{bmatrix}.
\]
The results are as follows:
\begin{itemize}
    \item \textbf{Z3-Complete-SR}: \( V(x) = x_1^2 + x_2^2 \).
    \item \textbf{LP-CEGIS}: \( V(x) = x_1^2/20 + 3 x_2^2/20 \).
    \item \textbf{LP-CEGIS-SR}: \( V(x) = x_1^2/2 + x_2^2/2 \).
    \item \textbf{Z3-CEGIS} fails to produce a valid Lyapunov function after 10 CEGIS steps.    
    \item \textbf{Z3-CEGIS-SR}: \( V(x) = x_1^2 + x_2^2 \). 
\end{itemize}
SOSTOOLS returns \( V(x) = 1.115x_1^2 + 4.603 \times 10^{-5} x_1 x_2 + 1.116x_2^2 \), which cannot be verified by Z3; this is clearly due to numerical approximations to an exact Lyapunov function.

\textbf{(E5)}
    Consider the Lorenz system 
\[
\dot{x} = \begin{bmatrix} -\sigma x_1 + \sigma x_2 \\ r x_1 - x_2 - x_1 x_3 \\ -b x_3 + x_1 x_2 \end{bmatrix},
\]
where \( \sigma = 10 \), \( r = 0.9999 \), and \( b = 8/3 \). It it known that this system is globally asymptotically stable for $r\le 1$. We found the following Lyapunov functions using various synthesis approaches:
\begin{itemize}
    \item \textbf{Z3-Complete-SR}: \( V(x) = 10001 x_1^2 / 100000 + x_2^2 + x_3^2 \).
    \item \textbf{LP-CEGIS}: fails to produce a valid Lyapunov function after 10 CEGIS steps.
    \item \textbf{LP-CEGIS-SR}: \( V(x) = 21 x_1^2 / 10 + 421 x_2^2 / 20 + 421 x_3^2 / 20 \).
    \item \textbf{Z3-CEGIS}: times out.    
    \item \textbf{Z3-CEGIS-SR}: \( V(x) = 63033 x_1^2/100 + 31933 x_2^2/5 + 31933 x_3^2/5 \). 
\end{itemize}
SOSTOOLS returns $V(x) = 0.3186x_1x_2 + 3.167x_2^2 + 3.167x_3^2$ (after omitting small terms), which can also be verified to be a global Lyapunov function.

\textbf{(E6)}
Consider the model \cite{sankaranarayanan2013lyapunov}
\[
\dot{x} = \begin{bmatrix} -x_1 + x_2^3 - 3 x_3 x_4 \\ -x_1 - x_2^3 \\ x_1 x_4 - x_3 \\ x_1 x_3 - x_4^3 \end{bmatrix}.
\]
\textbf{LP-CEGIS-SR} successfully finds the Lyapunov function 
\[
V(x) = \frac{x_1^2}{100} + \frac{x_2^4}{50} + \frac{x_3^2}{50} + \frac{x_4^2}{100}.
\] 
Other methods either time out or fail to return a solution after exhausting the CEGIS steps. SOSTOOLS finds 
a complicated expression that is difficult to verify.

\textbf{(E7)}
    Consider the system \cite{papachristodoulou2002construction}
\[
\dot{x} = \begin{bmatrix} 
-x_1^3 + 4 x_2^3 - 6 x_3 x_4 \\ 
-x_1 - x_2 + x_5^3 \\ 
x_1 x_4 - x_3 + x_4 x_6 \\ 
x_1 x_3 + x_3 x_6 - x_4^3 \\ 
-2 x_2^3 - x_5 + x_6 \\ 
-3 x_3 x_4 - x_5^3 - x_6 
\end{bmatrix}.
\]
\textbf{LP-CEGIS-SR} successfully finds the Lyapunov function 
\[
V(x) = x_1^2/100 + x_2^4/100 + x_3^2/50 + x_4^2/100 + x_5^4/200 + x_6^2/100.
\]
Other methods either time out or fail to return a solution after exhausting the CEGIS steps. The default setting of SOSTOOLs fails to find a quadratic or quartic Lyapunov function, although a restricted template may provide a quartic Lyapunov function as reported in \cite{papachristodoulou2002construction}. 

The following examples are used to illustrate the benefits of using symbolic reduction for synthesizing Lyapunov functions via LaSalle's principle.

\textbf{(E8)--(E10)}
    Consider the following three models:
\[
(E8)\, \dot{x} = \begin{bmatrix} x_2 \\ -x_1^3 - x_2^3 \end{bmatrix},
\;
(E9)\, \dot{x} = \begin{bmatrix} x_2 \\ -x_1^5 - 3 x_2 \end{bmatrix},
\;
(E10)\, \dot{x} = \begin{bmatrix} x_2 \\ -x_1 - 7 x_2^5 \end{bmatrix}.
\]
For each system, all approaches without the LaSalle loop failed to return a valid Lyapunov function. However, the LaSalle-enhanced methods successfully identified the following Lyapunov functions:
\begin{enumerate}
\item For the first system, Z3-Complete-SR-LaSalle found \( V(x) = x_1^4 + 2 x_2^2 \), LP-CEGIS-SR-LaSalle returned \( V(x) = x_1^4 / 100 + x_2^2 / 50 \), and Z3-CEGIS-SR-LaSalle computed \( V(x) = x_1^4 + 2 x_2^2 \).

\item For the second system, Z3-Complete-SR-LaSalle obtained \( V(x) = x_1^6 + 3 x_2^2 \), LP-CEGIS-SR-LaSalle produced \( V(x) = x_1^6 / 100 + 3 x_2^2 / 100 \), and Z3-CEGIS-SR-LaSalle found \( V(x) = x_1^6 + 3 x_2^2 \).

\item For the third system, Z3-Complete-SR-LaSalle discovered \( V(x) = x_1^2 + x_2^2 \), LP-CEGIS-SR-LaSalle returned \( V(x) = x_1^2 / 100 + x_2^2 / 100 \), and Z3-CEGIS-SR-LaSalle also found \( V(x) = x_1^2 + x_2^2 \).
\end{enumerate}
SOSTOOLS returns $V(x) = 0.7248x_1^2 - 8.155e^{-11}x_1x_2 + 0.7248x_2^2$, which is close to an exact Lyapunov function for the third system and fails to provide a valid Lyapunov functions for other two cases. 

\section{Results for Instability Testing}\label{app:instability}

\begin{example}[Instability]\label{ex:instability}
In our final experiment, we conduct ``speedy'' stability and instability tests using Z3-complete-SR on randomly generated systems. A total of 100 systems are randomly generated for each setup, with dimensions ranging from 2 to 10, with monomials constrained to a prescribed maximum degree and up to 3 terms per dimension. Systems identified as trivially unstable, i.e., those with linearizations containing eigenvalues with positive real parts, are excluded from trials. We also exclude vector fields with a zero component, because they cannot be asymptotically stable. A 1-second timeout is enforced for each call of Z3-complete-SR in both stability and instability tests. Table \ref{tab:instability}  reports the percentages of verified cases of stability and instability, along with the average synthesis time for both stable and unstable cases, encompassing the SMT solver time as well as the overhead from symbolic reduction. Notably, a substantial number of systems can be proven to be not GAS in under 1 second, even for dimensions up to 10. We also ran Z3-complete without symbolic reduction and observed a 100\% timeout rate for all cases, including two-dimensional systems. This highlights the advantages of symbolic reduction, enabling complete SMT synthesis for high-dimensional systems that would otherwise be infeasible, even for lower-dimensional cases.
\end{example}

\end{document}